\documentclass[a4paper,english]{article}
 
\usepackage{microtype}

\usepackage[utf8]{inputenc}
\usepackage{authblk}
\usepackage{booktabs,subcaption,dcolumn}
\usepackage{hyperref}
\usepackage{amsthm}
\usepackage{amsfonts}
\usepackage{fullpage}
\usepackage{authblk}
\usepackage{newtxtext,newtxmath}
\usepackage{tikz}
\usetikzlibrary{arrows}
\usetikzlibrary{patterns}
\usepackage{todonotes}
\usepackage{caption}
\usepackage{xspace}
\usepackage[linesnumbered, ruled, vlined]{algorithm2e}
\usepackage{standalone}
\usepackage{csquotes}
\usepackage[toc,page]{appendix} 

\theoremstyle{definition}
\newtheorem{theorem}{Theorem}
\newtheorem{observation}[theorem]{Observation}
\newtheorem{lemma}[theorem]{Lemma}
\newtheorem{definition}[theorem]{Definition}

\newcommand{\etal}{~et~al.\ }
\newcommand{\ie}{~i.e.,\ }
\newcommand{\cf}{cf.\ }

\newcommand{\B}{\{0,1\}}
\newcommand{\eps}{\varepsilon}
\newcommand{\f}{\boldsymbol{f}\xspace}
\newcommand{\bp}{\boldsymbol{p}\xspace}
\newcommand{\bq}{\boldsymbol{q}\xspace}
\newcommand{\br}{\boldsymbol{r}\xspace}
\newcommand{\bx}{\boldsymbol{x}\xspace}

\newcommand{\bd}{\boldsymbol{\delta}\xspace}
\newcommand{\poly}{\mathop{\rm poly}}

\newcommand{\N}{\mathbb{N}\xspace}
\newcommand{\R}{\mathbb{R}\xspace}
\newcommand{\Z}{\mathbb{Z}\xspace}

\newcommand{\problem}[1]{\ensuremath{\textsc{#1}}\xspace}

\newcommand{\EOML}{\problem{End-Of-Metered-Line}}
\newcommand{\EOMLshort}{\problem{EOML}}

\newcommand{\PLCPshort}{\problem{PLCP}}
\newcommand{\SARRIVAL}{\problem{S-Arrival}}
\newcommand{\ARRIVAL}{\problem{Arrival}}
\newcommand{\CLOPT}{\problem{CLOpt}}
\newcommand{\LOPT}{\problem{LocalOpt}}

\newcommand{\algo}[1]{\textsc{#1}\xspace}
\newcommand{\RUN}{\algo{Run}}
\newcommand{\Successor}{\algo{Successor}}
\newcommand{\Predecessor}{\algo{Predecessor}}
\newcommand{\Valuation}{\algo{Valuation}}

\newcommand{\class}[1]{\ensuremath{\mathsf{#1}}\xspace}

\newcommand{\FP}{\class{FP}}
\newcommand{\NP}{\class{NP}}
\newcommand{\coNP}{\class{coNP}}
\newcommand{\UP}{\class{UP}}
\newcommand{\coUP}{\class{coUP}}
\newcommand{\TFNP}{\class{TFNP}}
\newcommand{\FPSPACE}{\class{FPSPACE}}
\newcommand{\CLS}{\class{CLS}}
\newcommand{\PLS}{\class{PLS}}
\newcommand{\UPcoUP}{\ensuremath{\UP\cap\coUP}\xspace}
\newcommand{\NPcoNP}{\ensuremath{\NP\cap\coNP}\xspace}
\newcommand{\NL}{\class{NL}}

\newif\ifarxiveversioninternalmacro
\arxiveversioninternalmacrotrue 

\title{ARRIVAL: Next Stop in CLS\footnote{The research leading to these results has received funding from the European Research Council under the European Union's Seventh Framework Programme (FP/2007-2013) / ERC Grant Agreement n.~616787 and was supported by the project 17-09142S of GA \v{C}R. P.~Hubáček and V. Slívová were partially supported by PRIMUS grant PRIMUS/17/SCI/9. T.D. Hansen was partially supported by BARC which is funded by a VILLUM Investigator grant (16582) to Mikkel Thorup.}}

\author[1]{Bernd G\"artner}
\author[2]{Thomas Dueholm Hansen}
\author[3]{Pavel Hubáček}
\author[3]{Karel Král}
\author[4]{Hagar Mosaad}
\author[3]{Veronika Slívová}
\affil[1]{Department of Computer Science, ETH Z\"urich, Switzerland\protect\\
  \texttt{gaertner@inf.ethz.ch}}
\affil[2]{Department of Computer Science, University of Copenhagen, Denmark\protect\\
  \texttt{dueholm@di.ku.dk}}
\affil[3]{Computer Science Institute, Charles University, Prague, Czech Republic\protect\\
\texttt{\{hubacek, kralka, slivova\}@iuuk.mff.cuni.cz}}
\affil[4]{Department of Computer Science and Engineering, German University in Cairo, Egypt\protect\\
	\texttt{hagar.omar@student.guc.edu.eg}}

\begin{document}

\maketitle

\begin{abstract}
We study the computational complexity of \ARRIVAL, a zero-player game on $n$-vertex switch graphs introduced by Dohrau, G{\"{a}}rtner, Kohler, Matou{\v{s}}ek, and Welzl.
They showed that the problem of deciding termination of this game is contained in \NPcoNP.
Karthik~C.~S. recently introduced a search variant of \ARRIVAL and showed that it is in the complexity class~\PLS.
In this work, we significantly improve the known upper bounds for both the decision and the search variants of \ARRIVAL.

First, we resolve a question suggested by Dohrau\etal and show that the decision variant of \ARRIVAL is in \UPcoUP.
Second, we prove that the search variant of \ARRIVAL is contained in~\CLS. Third, we give a randomized $\mathcal{O}(1.4143^n)$-time algorithm to solve both variants.

Our main technical contributions are (a) an efficiently verifiable characterization of the unique witness for termination of the \ARRIVAL game, and (b) an efficient way of
sampling  from the state space of the game. We show that the problem of finding the unique witness is contained in \CLS, whereas it was previously conjectured to be
\FPSPACE-complete. The  efficient sampling procedure yields the first algorithm for the problem that has expected runtime $\mathcal{O}(c^n)$ with $c<2$.
\end{abstract}

\section{Introduction}

Variants of switch graphs have applications and are studied for example in combinatorics and in automata theory (\cf \cite{KatzRW12} and the references therein).
Dohrau\etal\cite{Dohrau2017} introduced \ARRIVAL, a natural computational problem on switch graphs, which they informally described as follows:

\blockcquote{Dohrau2017}{
	Suppose that a train is running along a railway network, starting from a designated origin, with the goal of reaching a designated destination. The network, however, is of a special nature: every time the train traverses a switch, the switch will change its position immediately afterwards. Hence, the next time the train traverses the same switch, the other direction will be taken, so that directions alternate with each traversal of the switch.

Given a network with origin and destination, what is the complexity of deciding whether the train, starting at the origin, will eventually reach the destination?
}

The above rather straightforward question remains unresolved.
Dohrau~et~al.~\cite{Dohrau2017} showed that deciding \ARRIVAL is unlikely to be
\NP-complete (by demonstrating that it is in \NPcoNP), but it is currently not known to be efficiently solvable.

To determine whether the train eventually reaches its destination, it is natural to consider a \emph{run profile},\ie the complete transcript describing how many times the train traversed each edge.
Dohrau\etal\cite{Dohrau2017} presented a natural integer programming interpretation of run profiles called \emph{switching flows}, which have the advantage of being trivial to verify.
The downside of switching flows is that they do not guarantee to faithfully represent the number of times each edge has been traversed; a switching flow might contain superfluous circulations compared to a valid run profile.
Nevertheless, Dohrau\etal\cite{Dohrau2017} proved that the existence of a switching flow implies that the train reaches its destination, and thus a switching flow constitutes an \NP witness for \ARRIVAL.

The \coNP  membership was shown by an insightful observation about the structure of switch graphs.
Specifically, the train reaches its destination if and only if it never enters a node from which there is no directed path to $ d $.
The railway network can thus be altered so that all such vertices point to an additional ``dead-end'' vertex $ \bar{d} $.
The \coNP witness is then simply a switching flow from the origin to the dead-end $ \bar{d} $.

Given that the decision variant of \ARRIVAL is in \NPcoNP, it is natural to study the search complexity of \ARRIVAL in the context of total search problems with the guaranteed existence of a solution,\ie within the complexity class \TFNP (which contains the search analogue of \NPcoNP).
Total search problems are classified into subclasses of \TFNP using the methodology proposed by Papadimitriou~\cite{Papadimitriou94} that clusters computational problems according to the type of argument assuring the existence of a solution.
Karthik~C.~S.~\cite{Karthik17} noticed that the search for a switching flow is a prime candidate to fit into the hierarchy of \TFNP problems.
He introduced \SARRIVAL, a search version of \ARRIVAL that seeks a switching flow to either the destination  $d$ or the dead-end vertex $\overline{d}$, and showed that it is contained in the complexity class \PLS~\cite{JohnsonPY88} of total problems amenable to local search.

Fearnley\etal\cite{FearnleyGMS17a} recently studied multiple variants of reachability games on switch graphs and as one of their results gave a lower bound on the complexity of deciding ARRIVAL.
Specifically, they showed that \ARRIVAL is \NL-hard.

\subsection{Our Results}\label{sec:our-results}
One of the open problems suggested by Dohrau\etal\cite{Dohrau2017} was whether deciding the termination of \ARRIVAL is contained in \UPcoUP.
Recall that given a railway network with an origin~$ o $ and a destination~$ d $, the transcript of the route of the train captured in the run profile from~$ o $ to~$ d $ (if it exists) is unique.
We show that it is possible to efficiently decide whether a switching flow corresponds to a run profile, which provides a positive answer to the above question and places \ARRIVAL inside \UPcoUP. We similarly also improve the upper bound on the search complexity of \ARRIVAL:
We show that \SARRIVAL is contained in the complexity class \CLS.
Daskalakis and Papadimitriou~\cite{DaskalakisP11} introduced \CLS to classify problems that can be reduced to local search over continuous domains.
\CLS contains multiple important search problems such as solving simple stochastic games, finding equilibria in congestion games, and solving linear complementarity problems on P-matrices.
For all of these problems, as well as for \SARRIVAL, we currently do not have a polynomial time algorithm, and they are not known to be complete for some subclass of \TFNP.

We establish the containment in \CLS through a reduction to \EOML (\EOMLshort), a total search problem that was recently introduced by Hub\'{a}\v{c}ek and Yogev~\cite{HubacekY17} who also showed that it is in \CLS.
In \EOMLshort we are given a source in a directed graph with vertices of in-degree and out-degree at most one, and the task is to find a sink or a source different from the given trivial source.
The access to the graph is given locally via information about the successor and predecessor of each vertex together with its distance from the trivial source (for the formal definition see Definition~\ref{def:EOML}).

Our result makes it unlikely for \SARRIVAL to be \PLS-hard, which was one of the possibilities suggested by the containment in \PLS shown by Karthik~C.~S.~\cite{Karthik17}.
This is due to known black-box separations among subclasses of \TFNP~\cite{Morioka01,BureshM04}, which suggest that \CLS is a proper subclass of \PLS.
Note that our reduction from \SARRIVAL to \EOML results in instances with a significantly restricted structure: the \EOML graph consists only of a single path and many isolated vertices.
We believe that this structure may in future work be used to show that \SARRIVAL is contained in \FP.

Our reduction from \SARRIVAL to \EOML also implies that we can use an algorithm by Aldous~\cite{Aldous83} to solve \SARRIVAL. The algorithm is randomized and runs in $\mathcal{O}(2^{n/2} poly(n))$ expected time on switch graphs with $n$ vertices. This is the first algorithm with expected runtime $\mathcal{O}(c^n)$ for $c < 2$. (A trivial $\mathcal{O}(2^{n} poly(n))$ time algorithm can be obtained by following the path of the train through the network.) Aldous' algorithm, in fact, solves any problem in \PLS. It samples a large number of candidate solutions and then performs a local search from the best sampled solution. The advantage of our reduction is that the resulting search space for \EOML is small enough to make Aldous' algorithm useful, unlike in the previous reduction by Karthik~C.~S.~\cite{Karthik17} that showed containment in \PLS.

Fearnley\etal\cite{FearnleyGMS17} recently gave a reduction from $P$-matrix linear complementarity problems (\PLCPshort) to \EOML. As in our case for \ARRIVAL, this implies that Aldous' algorithm can be used to solve \PLCPshort. In fact this gives the fastest known randomized algorithm for \PLCPshort, running in expected time $\mathcal{O}(2^{n/2} poly(n))$ for input matrices of dimension $n \times n$. Fearnley\etal do not make this observation themselves, but it is straightforward to check that their reduction also gives an efficient representation of the search space. Although Aldous' algorithm is very simple, it thus non-trivially improves the best runtime of algorithms for multiple problems. We believe that this way of applying Aldous' algorithm is a powerful technique that will produce additional results in the future.

\subsection{Technical remarks}\label{sec:technical}

Recall that a switching flow is a run profile with additional superfluous circulations compared to the valid run profile. Our main technical observation is a characterization of switching flows that correspond to the valid run profile. Given a switch graph $G$ and a switching flow $f$, we consider the subgraph $ G^* $ induced over the railway network by the ``last-used'' edges;
for every vertex $ v $, we include in $ G^* $ only the outgoing edge that was, according to the switching flow, used by the train last time it left from $ v $. Note that such last-used edges can be efficiently identified simply by considering the parity of the total number of visits at every vertex. When $f$ is a valid run profile, then it is straightforward to see that the subgraph $G^*$ is acyclic. We show that this property is in fact a characterization,\ie any switching flow for which the induced graph $ G^* $ is acyclic must be a run profile.
Given that this property is easy to check, we can use it to efficiently verify run profiles as \UP witnesses. (The \coUP witness is then a run profile to the dead-end at $ \bar{d} $.)

For our reduction from \SARRIVAL to \EOML we extend the above observation to partial switching flows that are not required to end at the destination. The vertices of the \EOML graph created by our reduction correspond to partial switching flows in the \SARRIVAL instance. The directed edges connect partial run profiles to their natural successors and predecessors, i.e., the partial run extended or shortened by a single step of the train.
Any switching flow that does not correspond to some partial run profile is an isolated vertex in the \EOML graph.
Finally, the trivial source is the empty switching flow, and the distance from it can be computed for any partial run simply as the number of steps taken by the train so far.
Given that there is only a single path in the resulting \EOML graph and that its sink is exactly the complete run, we get that the unique solution to the \EOML instance gives us a solution for the original instance of \SARRIVAL.

To make the reduction efficiently computable, we need to address the verification of partial run profiles.
As it turns out, partial run profiles can be efficiently verified using the graph $ G^* $, in a similar way to complete run profiles discussed above.
The main difference is that the graph of last-used edges for a partial run profile can contain a cycle, as the train might visit the same vertex multiple times on its route to the destination.
However, we show that there is at most one cycle in $ G^* $, which always contains the current end-vertex of the partial run.
The complete characterization of partial run profiles (which covers also full run profiles) is given in Lemma~\ref{lem:run-profile-characterization}, and the formal reduction is described in Section~\ref{sec:SARRIVAL-in-CLS}.

Finally, we show that every partial run profile is uniquely determined by its
last-used edges and its end-vertex. This limits the size of the search space for the
\EOMLshort instances that are produced by our reduction, which allows us to efficiently use Aldous' algorithm~\cite{Aldous83} to solve \ARRIVAL and \SARRIVAL.


\section{Preliminaries}\label{sec:prelim}

In the rest of the paper we use the following standard notation.
For $ k\in\N $, we denote by $ [k] $ the set $ \{1,\dots,k\} $.
For a graph $ G = (V,E) $, we reserve $n = |V|$ for the number of vertices.
The basic object that we study are switch graphs, as defined by Dohrau\etal\cite{Dohrau2017}.

\begin{definition}[switch graph]\label{def:switch-graph}
	A \emph{switch graph} is a tuple $G=(V,E,s_0,s_1)$ where $s_0, s_1\colon V \rightarrow V$ and $E=\{(v,s_0(v)),(v,s_1(v)) \mid \forall v \in V\}$.\footnote{Whenever $s_0(v) = s_1(v)$ for some vertex $v \in V$ we depict them as multiple edges in figures.}
\end{definition}

In order to avoid cumbersome notation, we slightly overload the use of $s_0,s_1$ and treat both as functions from vertices to edges; that is by $s_b(v)$ we denote the edge $(v,s_b(v))$ for~$b \in \B$.
We use this convention throughout the paper unless stated otherwise.

\begin{algorithm}[t] 	
	\small
	\SetKwInOut{Input}{Input}
	\SetKwInOut{Output}{Output}
	\Input{ a switch graph $G=(V,E,s_0,s_1)$ and two vertices $o,d\in V$ }
	\Output{ for each edge $ e\in E$, the number of times the train traversed $ e $ }
	  $v\gets o$ \hfill\tcp{position of the train}
		$\forall u \in V$ set $s\_\text{curr}[u] \gets s_0(u)$ and $s\_\text{next}[u] \gets s_1(u)$\\
		$\forall e \in E$ set $\br[e]\gets 0$ \hfill\tcp{initialize the run profile}
		$step \gets 0$\\
		\While{$v\neq d$}{
			$(v,w)\gets s\_\text{curr}[v]$\hfill\tcp{compute the next vertex}
			$\br[s\_\text{curr}[v]]$++\hfill\tcp{update the run profile}
			swap$(s\_\text{curr}[v],s\_\text{next}[v])$\\ 
			$v\gets w$ \hfill\tcp{move the train}
			$step \gets step + 1$\\
		}
		\Return $ \br $
	\caption{{\sc{Run}}}
	\label{alg:run} 
\end{algorithm}


The \ARRIVAL problem was formally defined by Dohrau\etal\cite{Dohrau2017} as follows.

\begin{definition}[\ARRIVAL~\cite{Dohrau2017}]\label{def:arrival}
Given a switch graph $G=(V,E,s_0,s_1)$  and two vertices~$o,d \in V$, the \ARRIVAL problem is to decide whether the algorithm {\sc{Run}} (Algorithm~\ref{alg:run}) terminates,\ie whether the train reaches the destination $ d $ starting from the origin $ o $.
\end{definition}

To simplify theorem statements and our proofs, we assume without loss of generality that both $s_0(d)$ and $s_1(d)$ end in $d$.

A natural witness for termination of the \RUN procedure considered in previous work~(e.g.,~\cite{Dohrau2017}) is a switching flow.
We extend the definition of a switching flow to allow for partial switching flows that do not necessarily end in the desired destination $ d $.

\begin{definition}[(partial) switching flow, end-vertex]\label{def:partialswitchingflow}
Let $G=(V,E,s_0,s_1)$ be a switch graph.
For $ o,d \in V $, we say that $\f \in \N^{2n}$ is a \emph{switching flow from $o$ to $d$} if the following two conditions hold.
 \begin{description} 
\item[Kirchhoff's Law (flow conservation):]\label{cond:kirchhoff}
$$\forall v \in V \colon \sum_{e=(u,v) \in E}
\f_e - \sum_{e=(v,w) \in E} \f_e = [v = d] - [v = o]\ ,$$
where $[\cdot]$ is the indicator variable of the event in brackets.
\item[Parity Condition:]\label{cond:parity}
$$\forall v \in V\colon \f_{s_1(v)}\leq
\f_{s_0(v)} \leq \f_{s_1(v)}+1\ .$$
 \end{description} 
Kirchoff's law means that $o$ emits one unit of flow, $d$ absorbs one unit of flow, and at all other vertices, in-flow equals out-flow. If $d=o$, we have a circulation.

Given an instance $(G= (V,E,s_0,s_1),o,d)$ of {\sc Arrival}, we say that $\f$ is a \emph{switching flow} if it is a switching flow from $o$ to $d$.
A vector $\f \in \N^{2n}$ is called a \emph{partial switching flow} iff $\f$  is a switching flow from $o$ to $v$ for some vertex $v \in V$.
We say that $v$ is the \emph{end-vertex} of the partial switching flow.
We denote the end-vertex of $\f$ by $v_{\f}$.
\end{definition}

\begin{definition}[(partial) run profile]\label{def:run-profile}
A \emph{run profile} is the switching flow $ \br $ returned by the algorithm \RUN (Algorithm~\ref{alg:run}) upon termination.
A \emph{partial run profile} is a partial switching flow corresponding to some intermediate value of $\br$ in the algorithm \RUN (Algorithm~\ref{alg:run}).
\end{definition}

\begin{observation}[Dohrau\etal{\cite[Observation 1]{Dohrau2017}}]
Each (partial) run profile is a (partial) switching flow.
\end{observation}

\begin{observation}
	An end-vertex $v_{\f}$ of a switching flow $\f$ is computable in polynomial time.
	\label{obs:endVertexPoly}
\end{observation}

\begin{proof}
	It is sufficient to determine which vertex has a net in-flow of one.
\end{proof}

\section{The Complexity of Run Profile Verification}\label{sec:run-profile-verification}

Dohrau\etal\cite{Dohrau2017} proved that it is possible to efficiently verify whether a given vector is a switching flow.
In this section we show that we can also efficiently verify whether a switching flow is a run profile.
Combining this with the results by Dohrau\etal\cite{Dohrau2017}, we prove that the decision problem of \ARRIVAL is in \UPcoUP (see Section~\ref{sec:arrival-decision-complexity}) and that the search problem of \ARRIVAL lies in the complexity class \CLS (see Section~\ref{sec:arrival-search-complexity}).
As outlined in Section~\ref{sec:technical}, our approach for verification of run profiles is based on finding a cycle in a natural subgraph of the railway network~$G$ defined below.
Specifically, we consider the subgraph of $ G $ that contains only the last visited outgoing edge of each vertex, i.e., every vertex has out-degree at most one (see Figure~\ref{fig:GLast} for an illustration).

\begin{definition}[$G_{\f}^*$]\label{def:graphLast}
  Let $(G = (V,E,s_0,s_1),o,d)$ be an instance of \ARRIVAL, and let
  $\f\in \N^{2n}$
  be a partial switching flow.
We define a graph $G_{\f}^*=(V,E^*)$ as follows
\begin{align*}
E^* = &\left\{ s_0(v) \colon \forall v \in V \text{ s.t. } \f(s_0(v))\neq \f(s_1(v))  \right\}\cup\\
&{} \left\{ s_1(v) \colon \forall v \in V \text{ s.t. } \f(s_0(v))=\f(s_1(v)) > 0\right\}.
\end{align*}
\end{definition}


\begin{observation}
	Given a partial switching flow $\f$, the graph $G_{\f}^*$ can be computed in polynomial time.
	\label{obs:lastEdgesPoly}
\end{observation}

\begin{lemma}\label{lem:run-profile-characterization}
	A partial switching flow $\f$ is a partial run profile iff $\f_{s_0(d)} = \f_{s_1(d)} = 0$ and one of the following two conditions holds:
	 \begin{enumerate} 
		\item There exists no cycle in $G_{\f}^*$.
		\item There exists exactly one cycle in $G_{\f}^*$ and this cycle contains the end-vertex of $\f$.
	 \end{enumerate} 
\end{lemma}

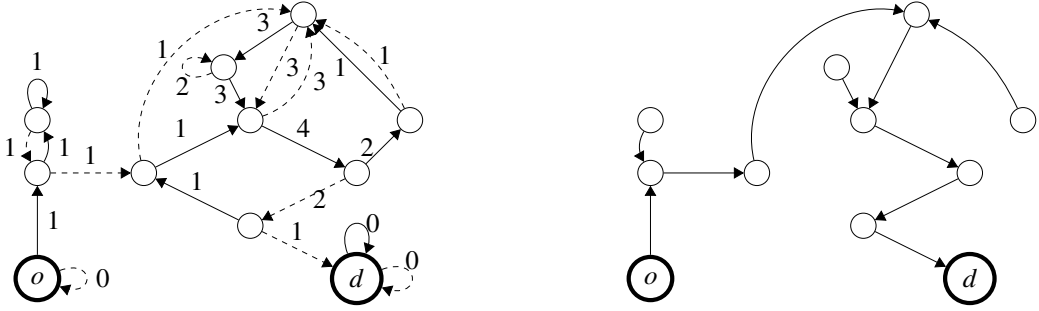
\begin{figure}[t]
	\centering
	\begin{minipage}{0.49\linewidth}
		\centering
		\IfFileExists{./pictures/priklad_boosted_flow.pdf}{
		\includestandalone[mode=image, scale=0.6]{pictures/priklad_boosted_flow}
	}{\scalebox{0.7}{\begin{tikzpicture}

		\Large
		\tikzset{vertex/.style = {shape=circle,draw,minimum size=0.8em}}
		\tikzset{edge/.style = {-triangle 60}}
		\tikzset{edge2/.style = {-triangle 60, dashed}}

		\node[style = {shape=circle,line width=0.80mm,draw,minimum size=1.3em}] (o) at    (2,-2) {$o$};
		\node[style = {shape=circle,line width=0.80mm,draw,minimum size=1.3em}] (d) at    (8,-2) {$d$};
		\foreach \nodename/\x/\y in { v_1/2/0, v_2/2/1, v_3/4/0, v_4/5.5/2, v_5/6/1, v_6/6/-1, v_7/7/3, v_8/8/0, v_9/9/1} {
			\node[vertex] (\nodename) at    (\x,\y) {};
		}
		\draw[edge] (o) to node[right] {1} (v_1);
		\draw[edge2] (o) to[out=20,in=340,looseness=6] node[right] {0} (o);

		\draw[edge] (v_1) to[bend right] node[right] {1} (v_2);
		\draw[edge2] (v_1) to node[above] {1} (v_3);

		\draw[edge] (v_2) to[out=115,in=65,looseness=10] node[above] {1} (v_2);
		\draw[edge2] (v_2) to[bend right] node[left] {1} (v_1);

		\draw[edge] (v_3) to node[above left] {1} (v_5);
		\draw[edge2] (v_3) to[bend left=58] node[left] {1} (v_7);

		\draw[edge] (v_4) to node[left] {3} (v_5);
		\draw[edge2] (v_4) to[out=205,in=155,looseness=10] node[below] {2} (v_4);

		\draw[edge] (v_5) to node[above] {4} (v_8);
		\draw[edge2] (v_5) to[bend right=45] node[right] {3} (v_7);

		\draw[edge] (v_6) to node[above] {1} (v_3);
		\draw[edge2] (v_6) to node[above] {1} (d);

		\draw[edge] (v_7) to node[above] {3} (v_4);
		\draw[edge2] (v_7) to node[right] {3} (v_5);

		\draw[edge] (v_8) to node[left] {2} (v_9);
		\draw[edge2] (v_8) to node[right] {2} (v_6);		

		\draw[edge] (v_9) to node[left] {1} (v_7);	
		\draw[edge2] (v_9) to[bend right=20] node[right] {1} (v_7);	

		\draw[edge] (d) to[out=110,in=70,looseness=6] node[right] {0} (d);	
		\draw[edge2] (d) to[out=20,in=340,looseness=6] node[above] {0} (d);

	
\end{tikzpicture}}}
	\end{minipage}
	\hfill
	\begin{minipage}{0.49\linewidth}
		\centering
		\IfFileExists{./pictures/priklad_boosted_gl.pdf}{
		\includestandalone[mode=image, scale=0.6]{pictures/priklad_boosted_gl}
	}{\scalebox{0.7}{	\begin{tikzpicture}

		\Large
		\tikzset{vertex/.style = {shape=circle,draw,minimum size=0.8em}}
		\tikzset{edge/.style = {-triangle 60,> = latex'}}
		\tikzset{edge2/.style = {-triangle 60,> = latex', dashed}}

		\node[style = {shape=circle,line width=0.80mm,draw,minimum size=1.3em}] (o) at    (2,-2) {$o$};
		\node[style = {shape=circle,line width=0.80mm,draw,minimum size=1.3em}] (d) at    (8,-2) {$d$};
		\foreach \nodename/\x/\y in { v_1/2/0, v_2/2/1, v_3/4/0, v_4/5.5/2, v_5/6/1, v_6/6/-1, v_7/7/3, v_8/8/0, v_9/9/1} {
			\node[vertex] (\nodename) at    (\x,\y) {};
		}
		\draw[edge] (o) to node[right] {} (v_1);

		\draw[edge] (v_1) to node[above] {} (v_3);

		\draw[edge] (v_2) to[bend right] node[left] {} (v_1);

		\draw[edge] (v_3) to[bend left=58] node[left] {} (v_7);

		\draw[edge] (v_4) to node[left] {} (v_5);

		\draw[edge] (v_5) to node[above] {} (v_8);

		\draw[edge] (v_6) to node[above] {} (d);

		\draw[edge] (v_7) to node[right] {} (v_5);

		\draw[edge] (v_8) to node[right] {} (v_6);		

		\draw[edge] (v_9) to[bend right=20] node[right] {} (v_7);	

	
\end{tikzpicture}}}
	\end{minipage}
	\caption{An example of a switch graph $ G $ with a switching flow $\f$ (on the left) and the corresponding graph $G_{\f}^*$ (on the right). The $ s_0 $ edges are denoted by full arrows and the $ s_1 $ edges are denoted by dashed arrows. For each edge in $ G $, the switching flow $ \f $ is specified by the adjacent integer value.}
	\label{fig:GLast}
\end{figure}

\begin{figure}[t]
	\centering
	\begin{minipage}{0.32\linewidth}
		\IfFileExists{./pictures/priklad_cycle.pdf}{
		\includestandalone[mode=image, scale=0.6]{pictures/priklad_cycle}
	}{\scalebox{0.7}{	\begin{tikzpicture}

		\Large
		\tikzset{vertex/.style = {shape=circle,draw,minimum size=0.8em}}
		\tikzset{edge/.style = {-triangle 60,> = latex'}}
		\tikzset{edge2/.style = {-triangle 60,> = latex', dashed}}

		\node[style = {shape=circle,line width=0.80mm,draw,minimum size=1.3em}] (o) at    (4,2) {$o$};
		\node[style = {shape=circle,line width=0.80mm,draw,minimum size=1.3em}] (d) at    (4,0) {$d$};
		\foreach \nodename/\x/\y in { v_1/2/0, v_2/2/2} {
			\node[vertex] (\nodename) at    (\x,\y) {};
		}
		\draw[edge] (o) to node[right] {} (v_1);
		\draw[edge2] (o) to[loop right] node[right] {} (o);

		\draw[edge] (v_1) to node[right] {} (v_2);
		\draw[edge2] (v_1) to node[below] {} (d);

		\draw[edge] (v_2) to[loop above] node[above] {} (v_2);
		\draw[edge2] (v_2) to[bend right] node[left] {} (v_1);

		\draw[edge] (d) to[loop above] node[right] {} (d);	
		\draw[edge2] (d) to[loop right] node[above] {} (d);

	
\end{tikzpicture}}}
	\end{minipage}
	\begin{minipage}{0.21\linewidth}
		\centering
		\IfFileExists{./pictures/priklad_cycle_gl1.pdf}{
		\includestandalone[mode=image, scale=0.6]{pictures/priklad_cycle_gl1}
	}{\scalebox{0.7}{	\begin{tikzpicture}

		\Large
		\tikzset{edge/.style = {-triangle 60,> = latex'}}
		\tikzset{edge2/.style = {-triangle 60,> = latex', dashed}}

		\node[style = {shape=circle,line width=0.80mm,draw,minimum size=1.3em}] (o) at    (4,2) {$o$};
		\node[style = {shape=circle,line width=0.80mm,draw,minimum size=1.3em}] (d) at    (4,0) {$d$};
		\node[style = {shape=circle,draw,minimum size=0.8em}] (v_1) at    (2,0) {};
		\node[style = {pattern=north east lines, pattern color={black!40}, shape=circle,draw,minimum size=0.8em}] (v_2) at    (2,2) {};
		\draw[edge] (o) to node[right] {} (v_1);

		\draw[edge] (v_1) to node[right] {} (v_2);

		\draw[edge] (v_2) to[loop above] node[above] {} (v_2);


	
\end{tikzpicture}}}
	\end{minipage}
	\begin{minipage}{0.21\linewidth}
		\centering
		\IfFileExists{./pictures/priklad_cycle_gl2.pdf}{
		\includestandalone[mode=image, scale=0.6]{pictures/priklad_cycle_gl2}
	}{\scalebox{0.7}{	\begin{tikzpicture}

		\Large
		\tikzset{edge/.style = {-triangle 60,> = latex'}}
		\tikzset{edge2/.style = {-triangle 60,> = latex', dashed}}

		\node[style = {shape=circle,line width=0.80mm,draw,minimum size=1.3em}] (o) at    (4,2) {$o$};
		\node[style = {shape=circle,line width=0.80mm,draw,minimum size=1.3em}] (d) at    (4,0) {$d$};
		\node[style = {pattern=north east lines, pattern color={black!40}, shape=circle,draw,minimum size=0.8em}] (v_1) at    (2,0) {};
		\node[style = {shape=circle,draw,minimum size=0.8em}] (v_2) at    (2,2) {};
		\draw[edge] (o) to node[right] {} (v_1);

		\draw[edge] (v_1) to node[right] {} (v_2);

		\draw[edge] (v_2) to[bend right] node[left] {} (v_1);


	
\end{tikzpicture}}}
	\end{minipage}
	\begin{minipage}{0.21\linewidth}
		\centering
		\IfFileExists{./pictures/priklad_cycle_gl3.pdf}{
		\includestandalone[mode=image, scale=0.6]{pictures/priklad_cycle_gl3}
	}{\scalebox{0.7}{	\begin{tikzpicture}

		\Large
		\tikzset{vertex/.style = {shape=circle,draw,minimum size=0.8em}}
		\tikzset{edge/.style = {-triangle 60,> = latex'}}
		\tikzset{edge2/.style = {-triangle 60,> = latex', dashed}}

		\node[style = {shape=circle,line width=0.80mm,draw,minimum size=1.3em}] (o) at    (4,2) {$o$};
		\node[style = {shape=circle, pattern=north east lines, pattern color={black!40}, line width=0.80mm,draw,minimum size=1.3em}] (d) at    (4,0) {$d$};
		\foreach \nodename/\x/\y in { v_1/2/0, v_2/2/2} {
			\node[vertex] (\nodename) at    (\x,\y) {};
		}
		\draw[edge] (o) to node[right] {} (v_1);

		\draw[edge] (v_1) to node[below] {} (d);

		\draw[edge] (v_2) to[bend right] node[left] {} (v_1);


	
\end{tikzpicture}}}
	\end{minipage}
	\caption{An example of a switch graph $G$ and cycles in the graphs $G^*$ corresponding to partial run profiles after 3, 4, and 5 steps of the train (respectively from left to right). We use hatching to highlight the current end-vertex.}
	\label{fig:changing_Glast}
\end{figure}
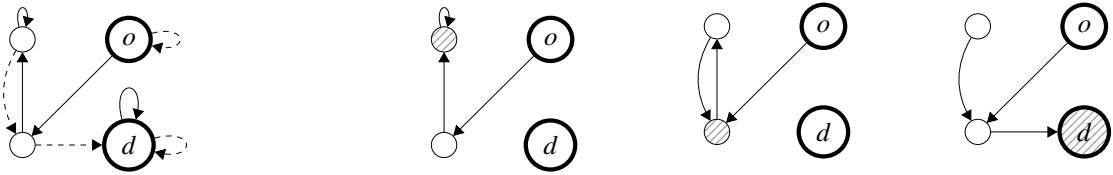

The main idea of the proof is based on the following fact:
a switching flow $f$ which is \emph{not} a run profile must contain a circulation (as shown by Dohrau\etal\cite{Dohrau2017}).
Let $f$ be a switching flow that we get from a run profile~$r$ by adding some flows on cycles,
then the last added circulation (the last added cycle) must form a cycle in the corresponding graph $G^*_f$.
On the other hand, a cycle containing the end-vertex is formed in $G^*_f$ whenever the train arrives to a previously visited vertex.
An illustration of the graph $ G^* $ at consecutive steps of the algorithm \RUN, with the corresponding evolution of the end-vertices and cycles, is given in Figure~\ref{fig:changing_Glast}.
\ifarxiveversioninternalmacro
	{

	}
		We prove the ``$\Leftarrow$'' implication of Lemma~\ref{lem:run-profile-characterization} by contradiction.
Given a switching flow $\f$ we consider the longest run profile $\br$ that is everywhere at most $\f$.
We denote the difference $\f - \br$ by $\bd$.
We utilize the following observation at two points in our proof.
First, to prove that both $ \f $ and $ \br $ have the same end-vertex, that is $v_{\br} = v_{\f}$.
Second, to prove that any cycle we find in the intersection of $G_{\f}^*$ and the non-zero edges of $\bd$ avoids $v_{\f}$.

\begin{observation}
Let $\f$ be a switching flow and let $\br$ be the longest partial run profile such that all coordinates are at most $\f$ and denote $\bd = \f - \br$.
Then $\bd_{s_0(v_{\br})} = \bd_{s_1(v_{\br})} = 0$.
\label{obs:trainEnd}
\end{observation}

\begin{proof}[Proof of Observation~\ref{obs:trainEnd}]
Suppose that either $\bd_{s_0(v_{\br})} \neq 0$ or $\bd_{s_1(v_{\br})} \neq 0$, and that $\br$ is such that the next edge for the train to continue on is $s_0(v_{\br})$.
From the maximality of $\br$ we get that $\bd_{s_0(v_{\br})}$ is equal to zero.
Then we get
\begin{align*}
\f_{s_1(v_{\br})} &> \br_{s_1(v_{\br})} \tag{since $\bd_{s_1(v_{\br})}$ is non-zero} \\
&= \br_{s_0(v_{\br})} \tag{from the parity condition for $\br$} \\
&= \f_{s_0(v_{\br})}  \tag{from the maximality of $\br$}
\end{align*}
This leads to a contradiction with the parity condition of Definition~\ref{cond:parity} as $\f_{s_1(v_{\br})} > \f_{s_0(v_{\br})}$.

The other case is similar: suppose that the train should continue with the edge $s_1(v_{\br})$.
From the maximality of~$\br$ we get that $\bd_{s_1(v_{\br})}$ is equal to zero.
And thus we get
\begin{align*}
\f_{s_0(v_{\br})} &> \br_{s_0(v_{\br})} \tag{since $\bd_{s_0(v_{\br})}$ is non-zero} \\
&= \br_{s_1(v_{\br})} + 1 \tag{from the parity condition for $\br$} \\
&= \f_{s_1(v_{\br})} + 1  \tag{from the maximality of $\br$}
\end{align*}
This gives a contradiction with the parity condition of Definition~\ref{cond:parity} as $\f_{s_0(v_{\br})} > \f_{s_1(v_{\br})} + 1$.
\end{proof}

\begin{proof}[Proof of Lemma~\ref{lem:run-profile-characterization}]
We start by proving that any partial run profile $\br$ has at most one cycle and that this cycle always contains the end-vertex $v_{\br}$.
We proceed by induction on the length of the run profile,\ie the number of steps in the algorithm \RUN.
The base case is the initial run profile $0^{2n}$ for which $G_{0^{2n}}^*$ contains no cycle.
For the induction step, let us assume that a vector $\br^i$ is a partial run profile and the vector $\br^{i+1}$ is the partial run profile after one step from $\br^{i}$.

If $G_{\br^i}^*$ contains a cycle then removing the edge from the end-vertex $v_{\br^i}$ to its successor in~$G_{\br^i}^*$ produces a cycle-free graph.
Adding an edge between $v_{\br^i}$ and $v_{\br^{i+1}}$ creates a graph with at most one cycle.
Moreover, if there is a cycle then it has to contain $v_{\br^{i+1}}$.

To prove the reverse implication, assume we have a partial switching flow~$\f$ that satisfies the conditions from the statement of this lemma,\ie $G_{\f}^*$ has at most one cycle and if the cycle is present it contains the end-vertex~$v_{\f}$.
Let us denote by $\br$ the longest partial run profile such that all its coordinates are at most $\f$, that is $\br_e \leq \f_e$ for each $e \in E$.
Consider the difference $\bd = \f - \br$ of the switching flow and its longest run profile.
We complete the proof by the following case analysis:

 \begin{enumerate} 
\item  If the end-vertex $v_{\f}$ is the same as the end-vertex $v_{\br}$ ($v_{\f} = v_{\br}$), then the difference $\bd = \f - \br$ is a flow, and moreover it is a circulation.
That is the Kirchhoff's law is satisfied in all vertices of $G$:
$$\forall v \in G\colon \sum_{e = (u,v) \in E} \bd_e - \sum_{e = (v,w) \in E} \bd_e = 0\ .$$
If $\bd$ is identically zero we are done as $\f = \br$, and thus $\f$ is a partial run profile.

Otherwise, we show that the circulation $\bd$ will result in a cycle in $G_{\f}^* $.
Namely we prove that $G_{\f}^*$ contains an outgoing edge from $v$ with a non-zero value in $\bd$ iff $v$ has a non-zero outgoing edge in $\bd$.
Using this and the fact that $\bd$ satisfies the Kirchhoff's law everywhere, we find a cycle in~$G_{\f}^*$.

Let $u$ be any vertex of $G$ such that only one of the edges $s_0(v)$ and $s_1(v)$ is non-zero in $\bd$.
We claim that then the non-zero edge is contained in the graph $G_{\f}^*$.
There are only two possible cases:
 \begin{enumerate} 

\item  Either $\br_{s_0(u)}=\f_{s_0(u)}$, and thus $\br_{s_0(u)} = \br_{s_1(u)}-1$ and $\f_{s_0(u)} = \f_{s_1(u)}$ (see Table~\ref{tab:1}), and $G_{\f}^*$ contains the edge $s_1(u)$,

\item  or $\br_{s_1(u)}=\f_{s_1(u)}$, and thus $\br_{s_0(u)}= \br_{s_1(u)}$ and $\f_{s_0(u)}  = \f_{s_1(u)}+1$ (see Table~\ref{tab:2}), and $G_{\f}^*$ contains the edge $s_0(u)$.

 \end{enumerate} 
On the other hand, if both $s_0(u)$ and $s_1(u)$ are non-zero in $\bd$ then $G_{\f}^*$ contains one of them.
Thus we may choose any non-zero edge in $\bd$ which is in  $G_{\f}^*$  and proceed by another adjacent non-zero directed edge (in~$\bd$) in $G_{\f}^*$ until we construct a directed cycle in $G_{\f}^*$.

By the Observation~\ref{obs:trainEnd}, we know that all outgoing edges from $v_{\br} = v_{\f}$ are zero in $\bd$, and thus the end-vertex~$v_{\f}$ is not contained in the cycle we have found.

\item
If the end-vertex $v_{\f}$ is not the same as the end-vertex $v_{\br}$ ($v_{\f} \neq v_{\br}$) then we get a contradiction with $\f$ being a partial switching flow.
It follows from Observation~\ref{obs:trainEnd} that
$$\sum_{e=(u,v_{\br}) \in E} \f_e - \sum_{e=(v_{\br},w) \in E} \f_e \geq \sum_{e=(u,v_{\br}) \in E} \br_e - \sum_{e=(v_{\br},w) \in E} \br_e = \begin{cases}
0 & v_{\br}=o,\\
1 & \text{otherwise,}\\
\end{cases}$$
which is in contradiction with $\f$ being a partial switching flow for which the end-vertex~$v_{\f}$ differs from~$v_{\br}$.
 \end{enumerate}  
 This concludes the proof of Lemma~\ref{lem:run-profile-characterization}.
\end{proof}

\begin{table}[t]
	\begin{subtable}{.49\linewidth}
		\begin{tabular}{|c|c|c|}
			\hline
				&  $s_0(u)$ & $s_1(u)$ \\
			\hline
				$\bd$ &  $0$ & $1$ \\
			\hline
				$\br$ &  $a$ & $a-1$ \\
			\hline
				$\f$ &  $a$ & $a$ \\
			\hline
		\end{tabular}
		\caption{Case (1.a).}
		\label{tab:1}
	\end{subtable}
	\begin{subtable}{.49\linewidth}
		\begin{tabular}{|c|c|c|}
			\hline
				&  $s_0(u)$ & $s_1(u)$ \\
			\hline
				$\bd$ &  $1$ & $0$ \\
			\hline
				$\br$ &  $a$ & $a$ \\
			\hline
				$\f$ &  $a+1$ & $a$ \\
			\hline
		\end{tabular}
		\caption{Case (1.b).}
		\label{tab:2}
	\end{subtable}
	\caption{Case analysis from the proof of Lemma~\ref{lem:run-profile-characterization} when only one of the edges $s_0(u)$ and $s_1(u)$ is non-zero in $\bd$.}
\end{table}

\else
	{}The full proof of Lemma~\ref{lem:run-profile-characterization} is provided in  Appendix~\ref{sec:run-profile-characterization-proof}.
\fi

\begin{lemma}\label{lem:run-profile-verification}
It is possible to verify in polynomial time whether a vector is a run profile.
\end{lemma}

\begin{proof}
	We can check that a vector $\f$ is a switching flow in polynomial time due to Dohrau~et~al.~\cite{Dohrau2017}.
	The construction of the graph $G_{\f}^*$ is polynomial by Observation~\ref{obs:lastEdgesPoly}.
	Lemma~\ref{lem:run-profile-characterization} gives us a polynomial time procedure to check if $\f$ is also a run profile as it is sufficient to check if $G_{\f}^*$ contains more than one cycle or whether it has a cycle not containing the end-vertex.
	This check can be done by a simple modification of the standard depth-first search on $G_{\f}^*$.
\end{proof}

\section{The Computational Complexity of \ARRIVAL}\label{sec:arrival-complexity}

In this section we use our efficient structural characterization of run profiles from Lemma~\ref{lem:run-profile-characterization} to improve the known results about the computational complexity of \ARRIVAL.
Specifically, we show that the decision version of \ARRIVAL is in \UPcoUP and the search version is in~\CLS.

\subsection{The Decision Complexity of \ARRIVAL}\label{sec:arrival-decision-complexity}

Our upper bound on the decision complexity of \ARRIVAL follows directly from the work of Dohrau\etal\cite{Dohrau2017} by application of Lemma~\ref{lem:run-profile-verification}.

\begin{theorem}\label{thm:upCapCoup}
	\ARRIVAL is in \UPcoUP.
\end{theorem}

\begin{proof}
The unique \UP certificate for a YES-instance of \ARRIVAL is the run profile $\br$ returned by the algorithm \RUN.
Clearly, for each YES-instance there exists only one such vector $\br$ and $\br$ does not exist for NO-instances.
By Lemma~\ref{lem:run-profile-verification}, we can determine whether a candidate switching flow $\br$ is a run profile in polynomial time.

The \coUP membership follows directly from the reduction of NO-instances of \ARRIVAL to YES-instances of \ARRIVAL as suggested by Dohrau\etal\cite{Dohrau2017}.
The reduction adds to the original graph $ G $ a new vertex $ \bar{d} $, and for each vertex $v\in V$ such that there is no directed path from $v$ to the destination $d$, the edges $s_0(v)$ and $s_1(v)$ are replaced with edges $(v,\bar{d})$.
This alteration of the original switch graph can be performed in polynomial time.
Dohrau~et~al.~\cite{Dohrau2017} proved that the train eventually arrives either at $d$ or $\bar{d}$.
The unique \coUP witness for \ARRIVAL is then a run profile from $ o $ to the dead-end $\overline{d}$.
\end{proof}

\subsection{The Search Complexity of \ARRIVAL}\label{sec:arrival-search-complexity}

The search complexity of \ARRIVAL was first studied by Karthik~C.~S.~\cite{Karthik17}, who introduced a total search variant of \ARRIVAL as follows.

\begin{definition}[\SARRIVAL~\cite{Karthik17}]\label{def:sarrival}
	Given a switch graph $G=(V,E,s_0,s_1)$ and a pair of vertices~$o,d \in V$, define a graph $G'$ as follows:
	 \begin{enumerate} 
		\item Add a new vertex $\bar{d}$.
		\item For each vertex $v$ such that there is no directed path from $v$ to $d$, replace edges $s_0(v)$ and $s_1(v)$ with edges $(v,\bar{d})$.
		\item Edges $s_0(d),s_1(d), s_0(\bar{d})$, and $s_1(\bar{d})$ are self-loops.
	 \end{enumerate} 
	The problem \SARRIVAL is to find a switching flow in $ G' $ either from $o$ to $d$ or from $o$ to $\bar{d}$.
\end{definition}

The above Definition~\ref{def:sarrival} is motivated by the proof of membership in \NPcoNP by Dohrau\etal\cite{Dohrau2017}.
Namely, in order to ensure that a solution for \SARRIVAL always exists, it was necessary to add to the switch graph $ G $ the dead-end vertex $ \bar{d} $.

Note that our method for efficient verification of run profiles from Lemma~\ref{lem:run-profile-verification} allows us to define a more natural version of \SARRIVAL directly on the graph~$G$ without any modifications.
Instead of relying on the dead-end vertices, we can use the fact that a partial run profile with an edge that was visited for $ 2^n + 1 $ times is an efficiently verifiable witness for NO-instances of \ARRIVAL.

\begin{definition}[\SARRIVAL\ - simplified]\label{def:simple-sarrival}
Given a switch graph $G=(V,E,s_0,s_1)$ and a pair of vertices $o,d \in V$, the \SARRIVAL problem asks us to find one of the following:
 \begin{enumerate} 
	\item a run profile $ \br\in \N^{[2n]} $ from $ o $ to $ d $, or
	\item a 
          run profile $ \br\in \N^{[2n]} $ from $ o $ to any $ v\in V $ such that
		 \begin{itemize} 
			\item $ \br_{(u,v)} = 2^n+1$, where $ u $ is the last vertex visited by the train before it reached the end-vertex $ v $ of $ \br $, and
			\item $ \br_{e'} \le 2^n $ for all $ e' \neq (u,v) $.
		 \end{itemize} 
 \end{enumerate} 
\end{definition}

The correspondence of the above version of \SARRIVAL to the original one follows formally from the following lemma.

\begin{lemma}[Karthik C.~S.~{\cite[Lemma 1]{Karthik17}}]\label{lem:runProfileUpperBound}
For any $G=(V,E,s_0,s_1)$ and a pair of vertices $o,d \in V$. Let $\br$ be a run profile (thus $v_{\br} = d$), then $\br_e \leq 2^n$ for each edge~$e \in E$.
\end{lemma}

To argue membership of our version of \SARRIVAL in \TFNP, we need to show that both types of solutions in Definition~\ref{def:simple-sarrival} can be verified efficiently.
Solutions of the first type are simply run profiles, and we have already shown that they can be verified in polynomial time in Lemma~\ref{lem:run-profile-verification}.
In order to be able to verify solutions of the second type, it remains to argue that for any partial run profile, the immediate predecessor of its end-vertex can be determined in polynomial time.

\begin{lemma}\label{lem:finding-predecessor}
	Let $\br$ be a partial run profile after $R \geq 1$ steps and $u$ be the vertex visited by the train at step $R-1$.
	Then
	 \begin{enumerate} 
		 \item either $u$ is the unique predecessor of $v_{\br}$ in $G_{\br}^*$, or
		\item  there is a single cycle in $G_{\br}^*$ containing $v_{\br}$ and $u$ is the predecessor of $v_{\br}$ on this cycle.
	 \end{enumerate} 
\end{lemma}

\begin{proof}
First, note that if $u$ is the end-vertex one step before $ v_{\br} $ becomes the end-vertex then~$G_{\br}^*$ must contain the edge $(u,v_{\br})$, as it is the last edge used by the train to leave $ u $.
Thus, in the first case the immediate predecessor of $v_{\br}$ in the partial run $ \br $ is unambiguously given by the only predecessor of~$v_{\br}$ in $G_{\br}^*$.

For the second case we show that $G_{\br}^*$ contains a directed cycle $C$ (containing the end-vertex $v_{\br}$) and $u$ is unambiguously given by the predecessor of $ v_{\br} $ in $G_{\br}^*$ that lies on $C$.
We find the cycle $C$ by constructing the longest possible directed path $c_0=v_{\br}, c_1, \ldots ,c_k$ in $ G_{\br}^* $ without repeating vertices.
Note that it cannot happen that $c_k$ has no outgoing edge in $G_{\br}^*$.
Otherwise, $\br$ would have two different end-vertices $ v_{\br} $ and $ c_k $ (as having no outgoing edge in~$G_{\br}^*$ means that the train has never left this vertex).
By Lemma~\ref{lem:run-profile-characterization}, the directed edge from $c_k$ has to end in the end-vertex $v_{\br}$, or else there would be a cycle in $G_{\br}^*$ that avoids $v_{\br}$.

The algorithm \RUN takes $R$ steps to generate the run profile $\br$,\ie  $\sum_{e \in E} \br_e = R$.
Let $t_{\br}\colon V \rightarrow \{0,1, \ldots , R-1\}$ be the function returning the last step after which a vertex was left by the train in the partial run profile $ \br $.
Observe that, except for the edge through which the train arrived to $v_{\br}$,\footnote{The inequality does not hold for $ v_{\br} $, since $t(v_{\br})$ has not been updated to the time $R$ yet.} it holds for all edges $(x,y) \in G_{\br}^*$ that
$
t_{\br}(x) < t_{\br}(x)+1 \leq t_{\br}(y).
$
However, the above inequality cannot hold for all edges on the cycle $C$, and thus $C$ has to contain the last used edge and the train had to be in~$c_k$ at step $R - 1$.
%
\end{proof}

\begin{observation}
	\SARRIVAL from Definition~\ref{def:sarrival} reduces to simplified \SARRIVAL from Definition~\ref{def:simple-sarrival}.
\end{observation}

\begin{proof}
	Given a solution of the second type of the simplified \SARRIVAL,\ie the long run profile, we can get a run profile $\br$ to $\bar{d}$ in polynomial time.
	For each vertex $u$ we can determine whether there is an oriented path from it to the destination $d$, and if there is no such path we set $\br_{s_0(v)} = \br_{s_1(v)} = 0$.
	We compute the end vertex $v_{\br}$ and set $s_0(v_{\br}) = 1$.
	All other components of $\br$ are set according to the original solution of the simplified \SARRIVAL.
\end{proof}

\subsubsection{\SARRIVAL is in \CLS}\label{sec:SARRIVAL-in-CLS}

Karthik~C.~S.~\cite{Karthik17} showed that \SARRIVAL is contained in the class \PLS.
We improve this result and prove that \SARRIVAL is in fact contained in \CLS. As a by-product, we also obtain a randomized algorithm for \SARRIVAL with runtime $\mathcal{O}(1.4143^n)$ which is the first algorithm for this problem with expected runtime $\mathcal{O}(c^n)$ for $c<2$.

The class of total search problems that are amenable to ``continuous'' local search was defined by Daskalakis and Papadimitriou~\cite{DaskalakisP11} using the following canonical problem.

\begin{definition}[\CLS~\cite{DaskalakisP11}]\label{def:CLS}
$\CLS$ is the class of total search problems reducible to the following problem called $\CLOPT$.

Given two arithmetic circuits $f\colon[0,1]^3\rightarrow[0,1]^3$ and $p\colon[0,1]^3\rightarrow[0,1]$, and two real constants $\eps,\lambda>0$, find either a point $x\in[0,1]^3$ such that $p(f(x)) \leq p(x)+\eps$ or a pair of points $x,x'\in[0,1]^3$ certifying that either $p$ or $f$ is not $\lambda$-Lipschitz.
\end{definition}

Instead of working with $ \CLOPT $, we use as a gateway for our reduction a problem called \EOML (\EOMLshort) which was recently defined and shown to lie in \CLS by Hubáček and Yogev~\cite{HubacekY17}.

\begin{definition}[\EOML]\label{def:EOML}
	Given circuits $S,P\colon \B^m\rightarrow\B^m$, and $V\colon\B^m \rightarrow  [2^m]\cup\{0\} $ such that $P(0^m) = 0^m \neq S(0^m)$ and $ V(0^m)=1$, find a string $x\in\B^m$ satisfying one of the following:
	 \begin{enumerate} 
		\item either $P(S(x)) \neq x$ or $S(P(x)) \neq x \neq 0^m$,
		\item $ x\neq 0^m $ and $V(x)=1$,
		\item either $V(x)> 0$ and $ V(S(x))-V(x)\neq 1$
		or $V(x)> 1$ and  $V(x)-V(P(x))\neq 1$.
	 \end{enumerate} 
\end{definition}

The circuits $ S,P $ from Definition~\ref{def:EOML} implicitly represent a directed graph with vertices labelled by binary strings of length $m$, where each vertex has both out-degree and in-degree at most one.
The circuit $P$ represents the predecessor and the circuit $S$ represents the successor of a given vertex as follows: there is an edge from a vertex~$u$ to a vertex $v$ iff $S(u) = v$ and~$P(v) = u$.
Finally, the circuit $ V $ can be thought of as an odometer that returns the distance from the trivial source at $ 0^m $ or value $ 0 $ for vertices lying off the path starting at the trivial source.
The task in \EOML is to find a sink or a source different from the trivial one at $0^m$ (the solutions of the second and of the third type in Definition~\ref{def:EOML} ensure that $ V $ behaves as explained above).

We are now ready to present our reduction from \SARRIVAL to \EOML.



\begin{theorem}\label{thm:cls}
	\SARRIVAL can be reduced to \EOML, and thus it is contained in \CLS.
\end{theorem}

\begin{proof}
Let $ (G,o,d) $ be an instance of \SARRIVAL.
We construct an instance of \EOMLshort that contains a vertex for each candidate partial switching flow over the switch graph $ G $,\ie for each vector with $2n$ coordinates and values from $[2^n+1]\cup \left\{ 0 \right\}$.
The \EOMLshort instance will comprise of a directed path starting at the initial (empty) partial run profile $0^{2n}$.
Each vertex on the path has an outgoing edge to its consecutive partial run profile.
Any vertex that does not correspond to a partial run profile becomes a self-loop.
Finally, the valuation circuit $V$ returns either the number of steps in the corresponding partial run profile or the zero value if the vertex does not correspond to a partial run profile.

Formal description of the circuits $S$, $P$, and $V$ defining the above \EOMLshort graph is given by algorithms \Successor (Algorithm~\ref{alg:successor}), \Predecessor (Algorithm~\ref{alg:predecessor}), and \Valuation (Algorithm~\ref{alg:valuation})
\ifarxiveversioninternalmacro
.
\begin{algorithm}[h!]
	\small
	\SetKwInOut{Input}{Input}
	\SetKwInOut{Output}{Output}
	\Input{ a vector $\bx\in ([2^n+1]\cup \left\{ 0 \right\})^{2n}$ }
	\Output{ a vector $\boldsymbol{y}\in ([2^n+1]\cup \left\{ 0 \right\})^{2n}$ }
	\eIf(\tcp*[f]{efficiently testable by Lemma~\ref{lem:run-profile-verification}}){$\bx$ is a partial run profile}{
			compute the end-vertex $v_{\bx}$\\
			\eIf{$v_{\bx}=d$ or $\bx_e= 2^n + 1$ for some $ e\in E $ }{
				\Return $\bx$\hfill\tcp{the train terminates or runs for too long}
			}{
				$ b \gets \bx(s_1(v_{\bx}))-\bx(s_0(v_{\bx}))$\hfill\tcp{parity of the current visit at $ v_{\bx} $}
				$e \gets s_{b}(v_{\bx})$ \hfill\tcp{the next edge to traverse}
				$\boldsymbol{y}_h \gets \begin{cases}
					\bx_h + 1 & \text{if } h= e\\
					\bx_h& \text{otherwise}\\
				\end{cases}$\hfill\tcp{run profile update}
				\Return $\boldsymbol{y}$
			}
		}
		{\Return $\bx$\hfill\tcp{self-loop}}
	\caption{\Successor}
	\label{alg:successor}
\end{algorithm}

\begin{algorithm}[h!]
	\small
	\SetKwInOut{Input}{Input}
	\SetKwInOut{Output}{Output}
	\Input{ a vector $\bx\in ([2^n+1]\cup \left\{ 0 \right\})^{2n}$ }
	\Output{ a vector $\boldsymbol{y}\in ([2^n+1]\cup \left\{ 0 \right\})^{2n}$ }
	\If{$ \bx = \boldsymbol{0}$}{
		\Return $\boldsymbol{0}$\hfill\tcp{the trivial source}
	}
	\If(\tcp*[f]{efficiently testable by Lemma~\ref{lem:run-profile-verification}}){$\bx$ is a partial run profile}{
		compute the end-vertex $v_{\bx}$\\
		\tcp{correctness by Lemma~\ref{lem:finding-predecessor}}
		\eIf{$v_{\bx}$ has a single predecessor in $G_{\bx}^*$}{
			$e \gets$  the only incoming edge of $v_{\bx}$ in $G_{\bx}^*$
		}
		{
			$e \gets $ the only incoming edge of $v_{\bx}$ which lies on a directed cycle
		}
		\If{ $\bx_e \le 2^n + 1$ and $\bx_{e'} < 2^n + 1$ for all $e'\neq e$}{
		$\boldsymbol{y}_h \gets \begin{cases}
		\bx_h - 1 & \text{if } h= e\\
		\bx_h& \text{otherwise}\\
		\end{cases}$\\
		\Return $\boldsymbol{y}$
		}
	}
	\Return $\bx$\hfill\tcp{self-loop}
	\caption{\Predecessor}
	\label{alg:predecessor}
\end{algorithm}

		\begin{algorithm}[h!]
	\small
	\SetKwInOut{Input}{Input}
	\SetKwInOut{Output}{Output}
	\Input{ a vector $\bx\in ([2^n+1]\cup \left\{ 0 \right\})^{2n}$ }
	\Output{ a value $v\in \N$ }
	$ v \gets 0 $\hfill\tcp{default value for self-loops}
	\If(\tcp*[f]{efficiently testable by Lemma~\ref{lem:run-profile-verification}}){$\bx$ is a partial run profile}{
		$v \gets 1 + \sum_{i=1}^{2n} \bx_i$
	}
	\Return $v$
	\caption{\Valuation}
	\label{alg:valuation}
\end{algorithm}

\else
	{} 	in Appendix~\ref{sec:appendixAlg}.
\fi
A polynomial bound on the size of the circuits $S, P$, and $V$ follows directly from Observation~\ref{obs:lastEdgesPoly} (computing~$G^*$), Lemma~\ref{lem:run-profile-verification} (testing whether a given vector is a partial run profile), Observation~\ref{obs:endVertexPoly} (computing the end-vertex), and Lemma~\ref{lem:finding-predecessor} (computing the previous position of the train).

Lemma~\ref{lem:run-profile-characterization} and Lemma~\ref{lem:finding-predecessor} imply that the \EOMLshort graph indeed consists of a single directed path and isolated vertices with self-loops.
By the construction of $ V $ (it outputs the number of steps of the train), there are no solutions of the second or the third type (\cf Definition~\ref{def:EOML}).
Thus, the \EOMLshort instance has a unique solution which has to correspond to a run profile in the original \SARRIVAL instance or to a partial run profile certifying that the train ran for too long (see the second type of solution in Definition~\ref{def:simple-sarrival}).
\end{proof}

\section{An $\mathcal{O}(1.4143^n)$ Algorithm for  \SARRIVAL}\label{sec:efficient-CLS}
Consider any problem that can be put into the complexity class \PLS,\ie can be reduced to the canonical \PLS-complete problem \LOPT (see also~\cite[Definition 1]{Karthik17}):

\begin{definition}[\LOPT]\label{def:LOPT}
	Given circuits $S\colon \B^m\rightarrow\B^m$, and $V\colon\B^m \rightarrow  [2^m]\cup\{0\} $, find a string $x\in\B^m$ such that $V(x) \ge V(S(x))$.
\end{definition}

Aldous~\cite{Aldous83} introduced the following simple algorithm that can be used to solve \LOPT:
pick $2^{m/2}$ binary strings uniformly and independently at random from $\B^m$, and let $x_{\max}$ be the selected string that maximizes the value $V(x)$.
Starting from $x = x_{\max}$, repeatedly move to the successor $S(x)$, until $V(x) \ge V(S(x))$.
He showed that the expected number of circuit evaluations performed by the algorithm before finding a local optimum is at most~$\mathcal{O}(m 2^{m/2})$. Note that in the case of \EOMLshort it is possible that all sampled solutions are isolated vertices in the \EOMLshort graph. In this case the $0^m$ string has the best known value, and the search is started from there.

In the case of \SARRIVAL, the \PLS membership proof of Karthik~C.~S. constructs the circuits for successor and valuation with $\mathcal{O}(n^2)$ input
bits~\cite[Theorem 2]{Karthik17}; the \EOMLshort instance constructed
in Theorem~\ref{thm:cls}---proving \CLS and in particular \PLS
membership---yields circuits of $\mathcal{O}(n^2)$ input bits as well. This
number is too high to yield a randomized algorithm of non-trivial runtime.

This number of $\mathcal{O}(n^2)$ input bits comes from the obvious encoding of
a partial run profile: each of the $2n$ edges has a nonnegative
integer flow value of at most $2^n+1$. But in fact, a
\emph{terminating} run has at most $n2^n$ partial run profiles, as no
\emph{vertex-state pair} $(v, s\_\text{curr})$ can repeat in
Algorithm~\ref{alg:run}. In other words, a partial run profile $\f$ is
determined by its end-vertex~$v_{\f}$ as well as the positions of all
switches at the time of the corresponding visit of $v_{\f}$. This
means that a partial run profile can be encoded with $n+\log_2 n$
bits, and if we had a \PLS or \CLS membership proof of \SARRIVAL\ with
circuits of this many inputs only, we could solve \SARRIVAL\ in time
$O(\poly(n)\times 2^{n/2})$.

Next, we show that such membership proofs indeed exist. For this, we show that the
above encoding (of a partial run profile by an end-vertex and the
positions of all switches) can be efficiently decoded: given an
end-vertex and the positions of all switches, we can efficiently
compute a unique candidate for a corresponding partial run
profile. The resulting encoding and decoding circuits can be composed
with the ones in Theorem~\ref{thm:cls} to obtain \PLS and \CLS
membership proofs with circuits of $n+\log_2 n$ input bits, and hence yield
a randomized algorithm of runtime
$O(\poly(n)\times 2^{n/2}) = O(1.4143^n)$, as explained above.

\subsection{Decoding Partial Run Profiles}
We work with instances of \SARRIVAL as in Definition~\ref{def:sarrival}, i.e., there is always a run profile either to $d$ or to $\bar{d}$.
This is without loss of generality~\cite{Dohrau2017,Karthik17}.

\begin{definition}
	Let $G=(V,E,s_0,s_1)$ be a switch graph, and let
	$o,d,\bar{d}\in V$ be as in Definition~\ref{def:sarrival}.
  The \emph{parity} of a run profile $\f\in\N^{2n}$ is the vector
  $\bp_{\f}\in\{0,1\}^{n-2}$ defined by
  $\bp_v=\f_{s_0(v)}-\f_{s_1(v)}\in\{0,1\}$, $v\in V\setminus\{d,\bar{d}\}$.
\end{definition}
Note that we do not care about
(the parity of) the switches at $d$ and $\bar{d}$, since the algorithm stops as soon as $d$ or $\bar{d}$ is reached.

Here is the main result of this section. For a given target vertex
$t\in V$ and given parity~$\bp$, there is exactly one candidate for a
partial run profile from $o$ to $t$ with parity
$\bp$. Moreover, this candidate can be computed by solving a system of
linear equations.

\begin{lemma}\label{lem:decode}
Let  $(G,o,d,\bar{d})$ be an instance
  of \SARRIVAL, let $t\in V$ and $\bp\in \{0,1\}^{n-2}$. Then there
  exists exactly one vector $\f \in \R^{2(n-2)}$ such that the following conditions hold.
 \begin{description} 
\item[Kirchhoff's Law (flow conservation):]\label{cond:kirchhoff2}
\begin{equation}\label{eq:kirchhoff2}
\forall v \in V\setminus\{d,\bar{d}\}\colon \sum_{e=(u,v) \in E}
\f_e - \sum_{e=(v,w) \in E} \f_e = [v = t] - [v = o]\,
\end{equation}
where $[\cdot]$ is the indicator variable of the event in brackets.
\item[Parity Condition:]\label{cond:parity2}
$\bp_{\f}=\bp$,\ie
\begin{equation}\label{eq:parity2}
\forall v \in V\setminus\{d,\bar{d}\}\colon \f_{s_0(v)}-\f_{s_1(v)} = \bp_v.
\end{equation}
 \end{description} 
\end{lemma}

Before we prove Lemma~\ref{lem:decode}, let us draw a crucial conclusion: The unique
partial run profile $\f\in\N^{2n}$ with end-vertex $t$ and
parity $\bp$ (if such $ f $ exists -- note that we are only guaranteed a real-valued $ \f $) necessarily satisfies (\ref{eq:kirchhoff2}) and
(\ref{eq:parity2}). Hence, we may use Lemma~\ref{lem:decode} to get
the entries $\f_e$ for all edges except the ones leaving~$d$ and~$\bar{d}$.
Only if all the entries are nonnegative and integral and
satisfy (\ref{eq:kirchhoff2}) at $d$ and $\bar{d}$ (under
$\f_{s_0(d)}=\f_{s_1(d)}=\f_{s_0(\bar{d})}=\f_{s_1(\bar{d})}=0$) do we have a candidate for a partial run profile. Hence, there is a
unique candidate, and given $t$ and a $\bp$, this candidate can
be efficiently found.

\begin{proof}[Proof of Lemma~\ref{lem:decode}]
  Set $m=2(n-2), V'=V\setminus\{d,\bar{d}\}$ and let
	$A\in\Z^{m\times m}$ be the coefficient matrix of the linear
  system (\ref{eq:kirchhoff2}), (\ref{eq:parity2}) in the variables
  $\f_e$. We show that $A$ is invertible.

	Let $\bq\in\R^{m}$ be the vector such that $\bq_{(v,s_i(v))}=-1$ if
  $s_i(v)= d$ and $\bq_{(v,s_i(v))}=0$ otherwise. We show that $\bq$ can be
  expressed as a linear combination of the rows of $A$ in a unique
  way, from which invertibility of $A$ and the statement of the lemma
  follow.

  Let us use coefficients $\lambda_v$ for each $v\in V'$ for the rows corresponding to the flow conservation constraints~(\ref{eq:kirchhoff2}), and coefficients $\mu_v$ for each $v\in V'$ for the rows corresponding to the parity constraints~(\ref{eq:parity2}).
  The column of $A$ corresponding to variable
	$\f_{(v,s_i(v))}$, has a $-1$ entry from the flow
  conservation constraint at $v$, and a $1$ entry (if $i=0$) or a $-1$
  entry (if $i=1$) from the parity constraint at $v$. If
  $s_i(v)\neq d,\bar{d}$, there is another $1$ entry from the flow
  conservation constraint at $s_i(v)$.  All other entries are
  zero. The equations that express $\bq$ as a linear combination of
  rows of $A$ are therefore the following.
	\begin{align*}
		&\forall v \in V': ~ - \lambda _v + \mu_v + \lambda_{s_0(v)} \cdot [s_0(v)\neq d,\bar{d}] ~=~ \bq_{(v,s_0(v))},\\
		&\forall v \in V': ~ - \lambda _v - \mu_v + \lambda_{s_1(v)} \cdot [s_1(v)\neq d,\bar{d}] ~=~ \bq_{(v,s_1(v))}.
	\end{align*}
	Or equivalently:
    \begin{equation}\label{eq:ssg1}
      \lambda_v - \mu_v = \left\{\begin{array}{ll}
                                   \lambda_{s_0(v)}, & \mbox{if
                                   $s_0(v)\neq d,\bar{d}$} \\
                                   1, &\mbox{if $s_0(v)=d$} \\
                                   0, & \mbox{if $s_0(v)=\bar{d}$}
                                 \end{array}\right.\quad v\in V',
     \end{equation}
  \begin{equation}\label{eq:ssg2}
      \lambda_v + \mu_v = \left\{\begin{array}{ll}
                                   \lambda_{s_1(v)}, & \mbox{if
                                   $s_1(v)\neq d,\bar{d}$} \\
                                   1, &\mbox{if $s_1(v)=d$} \\
                                   0, & \mbox{if $s_1(v)=\bar{d}$}
                                 \end{array}\right. \quad v\in V'.
     \end{equation}

    We now show that there are unique coefficients $\lambda_v,\mu_v$
    satisfying these equations. Let us define $\lambda_d=1$ and
    $\lambda_{\bar{d}}=0$. Adding corresponding equations of
    (\ref{eq:ssg1}) and (\ref{eq:ssg2}) then yields
		\[
		\lambda_d = 1, \quad \lambda_{\bar{d}} = 0, \quad \lambda_v = \frac{1}{2}\left(\lambda_{s_0(v)}+\lambda_{s_1(v)}\right), \quad \forall v\in V'.
		\]
    These are exactly the equations for the vertex values in a
    \emph{stoppping simple stochastic game} on the graph $G$ with only
    average or degree-1 vertices and sinks $d$ and $\bar{d}$ (stopping
    means that $d$ or $\bar{d}$ are reachable from everywhere which is
    exactly what we require in a switch graph). Condon proved
    that these values are unique~\cite{Con}. This also determines the
    $\mu_v$'s uniquely.
\end{proof}

\section{Conclusion and Open Problems}\label{sec:conclusion}

We showed that candidate run profiles in \ARRIVAL can be efficiently verified due to their structure.
This allowed us to improve the known upper bounds for the search complexity of \ARRIVAL and \SARRIVAL.
Here we mention some natural questions arising from our work.
 \begin{itemize} 
\item Are there any non-trivial graph properties that make \ARRIVAL or \SARRIVAL efficiently solvable?
Given that we currently do not know of any polynomial time algorithm for \ARRIVAL on general switch graphs, we could study the complexity of \ARRIVAL on some interesting restricted classes of switch graphs.
\item Are there other natural problems in \UPcoUP such that their corresponding search variant is reducible to \EOMLshort?
Does \EOML capture the computational complexity of any \TFNP problem with unique solution?
Fearnley\etal\cite{FearnleyGMS17} recently gave a reduction from the \PLCPshort to \EOMLshort.
Given that \ARRIVAL and \PLCPshort can be both reduced to \EOMLshort, yet another intriguing question is whether there exists any reduction between the two.
\item
  As mentioned in Section \ref{sec:our-results}, the reduction from \PLCPshort to \EOMLshort by Fearnley\etal\cite{FearnleyGMS17} implies that \PLCPshort can be solved faster with Aldous' algorithm~\cite{Aldous83} than with any other known algorithm. It would be interesting to see whether Aldous' algorithm can similarly give improved runtimes for other problems than \ARRIVAL and \PLCPshort.
\end{itemize} 

\subparagraph*{Acknowledgements.}
We wish to thank Karthik~C.~S. for helpful discussions and suggestions.

\bibliography{bibliography}


\end{document}